\documentclass[conference]{IEEEtran}
\IEEEoverridecommandlockouts
\PassOptionsToPackage{draft,bookmarks=false}{hyperref}

\usepackage{relsize}
\usepackage{cite}
\usepackage{amsmath,amssymb,amsfonts}
\usepackage{algorithmic}
\usepackage{graphicx}
\usepackage[colorlinks]{hyperref}
\usepackage[caption=false,font=footnotesize]{subfig}
\hypersetup{citecolor=blue}
\usepackage{textcomp}
\usepackage{xcolor}
\usepackage{tikz}
\usetikzlibrary{arrows.meta, fit}
\def\BibTeX{{\rm B\kern-.05em{\sc i\kern-.025em b}\kern-.08em
    T\kern-.1667em\lower.7ex\hbox{E}\kern-.125emX}}

\usepackage[font=small,skip=2pt]{caption}

\usepackage{enumitem}
\setlist{nosep}

 \newtheorem{proposition}{Proposition}

\newtheorem{remark}{Remark}
\begin{document}

%\IEEEaftertitletext{\vspace{-6mm}}
\title{Finite-Blocklength ISAC Multiple Access: \\
	A Source-Channel Coding Perspective
%	\vspace{-5mm}
\author{Zhentian Zhang, Kaitao Meng, Hao Jiang, Zaichen Zhang
	% <-this % stops a space
	\thanks{}% <-this % stops a space
	\thanks{Zhentian Zhang, and Zaichen Zhang are with the National Mobile Communications Research Laboratory, Southeast University, Nanjing, 210096, China and Hao Jiang is with the School of Cyber Science and Engineering, Southeast University, Nanjing 210096, P. R. China (e-mail: zhentianzhangzzt@gmail.com, zczhang.seu.edu.cn, jiang.hao@seu.edu.cn ).}% <-this % stops a space
	\thanks{Kaitao Meng is with the Department of Electrical and Electronic Engineering, University of Manchester, Manchester, UK (email: kaitao.meng@manchester.ac.uk).}
}
}

\maketitle
\begin{abstract}
	Future networks must serve massive populations of devices that sense and
	communicate simultaneously under short-packet constraints, yet the
	fundamental limits of integrated sensing and communication (ISAC) in the
	finite-blocklength multiple-access regime remain largely undiscovered. This
	paper closes this gap from a source-channel coding perspective. We prove
	that satisfying a sensing-distortion constraint is information-theoretically
	equivalent to a source-coding requirement, which collapses sensing and
	communication into the joint recovery of a single effective payload within a coded multiple-access framework. Building on this
	equivalence, we derive a finite-blocklength achievability bound together
	with a Fano-sum many-user converse and a genie-aided single-user converse,
	yielding a tight characterization of the minimum energy per bit and the
	rate-sensing tradeoff. Numerical results reveal that the energy price of
	sensing fidelity grows almost linearly in dB per decade of distortion
	tightening and is significantly amplified by the multiple-access load, and
	that joint encoding of the effective payload strictly outperforms an
	optimized orthogonal two-phase scheme, demonstrating a genuine integration
	gain of ISAC at finite blocklength.
\end{abstract}

\begin{IEEEkeywords}
	ISAC, finite-blocklength, sensing distortion, source coding, many-access 
	channel, multiple access, tradeoffs.
\end{IEEEkeywords}

\section{Introduction}
Integrated sensing and communication (ISAC) is reshaping wireless network
design and advancing cognitive networks \cite{ISAC1,ISAC2}. While
ISAC-oriented transceiver designs have demonstrated substantial integration
gains in downlink settings \cite{ISAC3}, attention has recently shifted to
the uplink, where receivers either jointly decode messages and retrieve
sensing information from the transmitted signals
\cite{Uplink_ISAC1,Uplink_ISAC2}, or sequentially decode one user's message
before sensing the state of another \cite{Uplink_ISAC3}. A prominent trend
in uplink ISAC is the transition toward message decoding-oriented sensing,
reflecting the practical reality that uplink systems operate with encoded
finite-blocklength transmissions and thus inherently call for
multiple-access design, particularly at the initial access stage. Existing
work has characterized single-user ISAC tradeoffs in which sensing
parameters are estimated from feedback signals after message decoding
\cite{FBL_ISAC2}, and more recently \cite{FBL_ISAC1} studied uplink ISAC
multiple access based solely on communication-signal-enabled sensing.
Nevertheless, a fundamental analytical framework for uplink-only,
feedback-free ISAC under varying {\em access density}, a setting that
faithfully captures practical uplink operation, remains absent, and it is
unclear how the sensing fidelity, the decoding reliability, and the energy
budget jointly scale when many short-packet users access the channel
simultaneously.

In this paper, we close this gap by revisiting the classical Gaussian
multiple-access channel (MAC) in the finite-blocklength regime and
reformulating uplink transmission as a joint sensing-and-communication
problem from a source-channel coding perspective, which enables a unified
treatment of reliability, sensing performance, and energy efficiency. Our
main contributions are summarized as follows.
\begin{itemize}
	\item We prove that distortion-constrained sensing over finite
	blocklengths is information-theoretically equivalent to a source-coding
	problem under both mutual-information and joint rate-distortion
	interpretations, which folds the sensing target into a
	sensing-equivalent payload.
	\item We formulate joint sensing-and-communication recovery in terms of
	an effective information load, yielding a tight characterization of
	energy efficiency, namely the minimum energy \cite{energy-bound}
	required to meet prescribed decoding-error and sensing-distortion
	targets, and rendering the rate-sensing tradeoff analytically
	tractable.
	\item We derive finite-blocklength achievability bounds together with
	single-user and many-user converse bounds, and demonstrate both
	analytically and numerically that they tightly sandwich the fundamental
	limits across a wide range of access densities.
\end{itemize}

\begin{figure}[t!]
	\centering
	\setlength{\fboxsep}{4pt}%
	\setlength{\fboxrule}{0.6pt}%
	\fbox{\resizebox{\dimexpr\columnwidth-2\fboxsep-2\fboxrule\relax}{!}{%
			\begin{tikzpicture}[font=\scriptsize, line width=0.55pt,
				arr/.style={-{Stealth[length=1.6mm]}},
				msg/.style={draw=black, fill=white,   rounded corners=1.5pt,
					align=center, inner sep=3pt, minimum height=6mm, minimum width=17mm},
				cmp/.style={draw=black, fill=gray!12, rounded corners=1.5pt,
					align=center, inner sep=3pt, minimum height=6mm},
				enc/.style={draw=black, fill=gray!12, rounded corners=1.5pt,
					align=center, inner sep=3pt, minimum height=11mm, minimum width=12mm},
				dec/.style={draw=black, fill=gray!12, rounded corners=1.5pt, thick,
					align=center, inner sep=3pt, minimum height=9mm},
				outbox/.style={draw=gray!70, fill=white, rounded corners=1.5pt,
					align=left, inner sep=3.5pt},
				usr/.style={draw=gray!60, dashed, rounded corners=2.5pt, inner sep=4pt},
				tag/.style={font=\tiny, text=gray!70!black, fill=white, inner sep=1.5pt}]
				
				% ================= user 1 =================
				\node[msg] (w1) at (0,2.85)    {message $w_1$};
				\node[msg] (t1) at (0,2.05)    {sensing $\Theta_1$};
				\node[cmp] (r1) at (1.80,2.05) {lossy\\[-1pt] compression};
				\node[enc] (e1) at (3.45,2.45) {Encoder $1$};
				\draw[arr] (t1) -- (r1);
				\draw[arr] (w1.east) -- (e1.west |- w1);
				\draw[arr] (r1.east) -- (e1.west |- r1);
				\node[usr, fit=(w1)(t1)(r1)] (g1) {};
				\node[tag] at ([xshift=7mm]g1.north west) {user $1$};
				
				% ================= user k_a =================
				\node[msg] (wK) at (0,0.40)     {message $w_{k_a}$};
				\node[msg] (tK) at (0,-0.40)    {sensing $\Theta_{k_a}$};
				\node[cmp] (rK) at (1.80,-0.40) {lossy\\[-1pt] compression};
				\node[enc] (eK) at (3.45,0.00)  {Encoder $k_a$};
				\draw[arr] (tK) -- (rK);
				\draw[arr] (wK.east) -- (eK.west |- wK);
				\draw[arr] (rK.east) -- (eK.west |- rK);
				\node[usr, fit=(wK)(tK)(rK)] (gK) {};
				\node[tag] at ([xshift=7mm]gK.north west) {user $k_a$};
				
				\node[gray] at (0.90,1.22) {$\vdots$};
				\node[gray] at (3.45,1.22) {$\vdots$};
				
				% ================= channel =================
				\node[draw=black, circle, inner sep=1.2pt] (sum) at (5.05,1.22) {$+$};
				\draw[arr] (e1.east) -- (sum);
				\draw[arr] (eK.east) -- (sum);
				\node[font=\scriptsize] (z) at (5.05,0.40) {noise};
				\draw[arr] (z) -- (sum);
				\node[font=\tiny\itshape, text=gray!70!black] at (5.05,0.02) {Gaussian MAC};
				
				% ================= joint decoder =================
				\node[dec] (dec) at (7.40,1.22) {Joint decoder\\[-1pt] \& estimator};
				\draw[arr] (sum) -- node[below=1pt, midway, font=\tiny] {$n$ channel uses} (dec);
				\node[outbox, anchor=west] (res) at (8.55,1.22)
				{decoded messages\\[-1pt] {\quad\footnotesize error prob.\ $\le\epsilon$}\\[2pt]
					estimated parameters\\[-1pt] {\quad\footnotesize distortion $\le\epsilon_s$}};
				\draw[arr] (dec) -- (res);
				
				% ================= key message =================
				\node[draw=gray!60, dashed, rounded corners=2pt,
				align=center, text width=9.0cm, inner sep=4pt] at (5.05,-1.75)
				{\emph{Source--channel coding view.} Each sensed parameter is
					compressed into a few extra bits and embedded in the user's
					message, so sensing and communication share a single
					finite-blocklength channel code.};
			\end{tikzpicture}%
	}}
	\caption{Block diagram of the considered finite-blocklength ISAC multiple-access system.}
	\label{fig:system_model}
\end{figure}
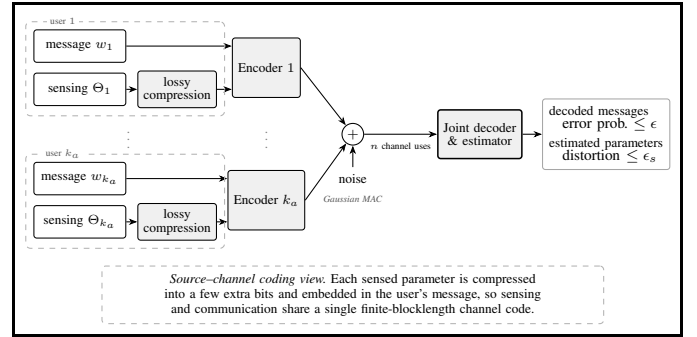

The remainder of this paper is organized as follows:
Section~\ref{sec.2} presents the information-theoretic system model.
Section~\ref{sec.standard_ka_jsac_mac} establishes the source-coding
equivalence and derives the achievability and converse bounds for the
standard $k_a$-user ISAC MAC. Section~\ref{sec.4} provides numerical
results, and Section~\ref{sec.5} concludes the paper.

\section{Information-Theoretic Modeling}\label{sec.2}

We consider an uplink joint communication-and-remote-sensing system in which $k_a$ active
users transmit over an $n$-dimensional real AWGN multiple-access channel (MAC). Each user
$k$ holds two pieces of information: a digital message $w_k\in\mathcal{W}$ carrying $b$
bits (so $|\mathcal{W}|=2^b$), and a locally acquired sensing parameter
$\Theta_k\in\mathcal{T}\subseteq\mathbb{R}^d$, drawn i.i.d.\ across users and independent
of the messages. Since the channel input of each user depends on both quantities, the
received signal is
\begin{equation}
	Y=\sum_{k=1}^{k_a} X_k(w_k,\Theta_k)+Z,
	\label{eq:received_signal}
\end{equation}
where $Y\in\mathbb{R}^n$, $X_k\in\mathbb{R}^n$ is the transmitted vector of user $k$
subject to the block-energy constraint $\|X_k\|_2^2\le p$, and
$Z\sim\mathcal{N}(0,I_n)$ is AWGN. The parameter $\Theta_k$ is known at encoder $k$ but
\emph{not} at the receiver. Under this convention, the communication-oriented
energy-per-bit ratio is 
\begin{equation}
	\frac{e_b}{n_0} = \frac{p}{2b}.
\end{equation}

Upon observing $Y$, the receiver must (i) recover each message $w_k$ reliably, and
(ii) reconstruct each $\Theta_k$ within a mean-squared-error target,
\begin{equation}
	\mathbb{E}\big[\|\Theta_k-\widetilde{\Theta}_k\|_2^2\big]\le \epsilon_s.
	\label{eq:sensing_distortion}
\end{equation}
The system is thus a \emph{joint lossless-message and lossy-source transmission} problem
over a MAC.

The remainder of this section establishes a single, simple takeaway: \emph{from a converse perspective, the sensing task behaves exactly like an additional digital payload}. Concretely, we show in three steps that
\begin{enumerate}
	\item Meeting the distortion target $\epsilon_s$ forces the receiver to acquire at least
	$b_s$ bits of information about $\Theta_k$, where $b_s$ is determined by rate-distortion
	theory (Proposition~\ref{prop1});
	\item Because the message and the sensing parameter are independent, the two information
	requirements \emph{add up}, yielding a total budget of $b+b_s$ bits per user
	(Proposition~\ref{prop2});
	\item This additivity is exactly the joint rate-distortion function of the pair
	$(w,\Theta)$. The quantity $b+b_s$ then serves as the effective per-user payload in the
	finite-blocklength analysis of the next section.
\end{enumerate}

\paragraph{Sensing Accuracy as a Source-Coding Requirement}\label{sec.sensing_bits}
The intuition is the following: if the receiver can reproduce $\Theta$ within distortion
$\epsilon_s$, then the channel output $Y$ must already ``contain'' as much information
about $\Theta$ as the best lossy compressor operating at that distortion, i.e., no estimator
can do better than what rate-distortion theory permits. 

Formally, for squared-error
distortion $\mathcal{D}_s(\Theta,\widetilde{\Theta})=\|\Theta-\widetilde{\Theta}\|_2^2$,
the minimum number of bits needed to describe $\Theta$ within average distortion
$\epsilon_s$ is the rate-distortion function~\cite{shannon1959,berger1971,cover2006}
\begin{equation}
	\mathcal{R}_{\Theta}(\epsilon_s)
	=
	\inf_{\mathbb{P}_{\widetilde{\Theta}|\Theta}:\,
		\mathbb{E}[\mathcal{D}_s(\Theta,\widetilde{\Theta})]\le \epsilon_s}
	\mathcal{I}(\Theta;\widetilde{\Theta}),
	\label{eq:rd_sensing}
\end{equation}
where $\mathcal{I}(\cdot;\cdot)$ is mutual information in bits. We call
$b_s\triangleq \mathcal{R}_{\Theta}(\epsilon_s)$ the \emph{sensing-equivalent number of
	bits}\,\eqref{eq:bs_def}.
% keep the label available:
\begin{equation}
	b_s\triangleq \mathcal{R}_{\Theta}(\epsilon_s).
	\label{eq:bs_def}
\end{equation}

\begin{proposition}[Sensing as a Source-Coding Requirement]\label{prop1}
	If an estimator $\widetilde{\Theta}=\mathcal{G}_s(Y)$ satisfies
	$\mathbb{E}\big[\|\Theta-\widetilde{\Theta}\|_2^2\big]\le \epsilon_s$, then
	$\mathcal{I}(\Theta;Y)\ge b_s$.
\end{proposition}

\begin{IEEEproof}
	Since $\widetilde{\Theta}$ is computed from $Y$ alone, $\Theta\to Y\to\widetilde{\Theta}$
	forms a Markov chain, and the data-processing inequality gives
	$\mathcal{I}(\Theta;Y)\ge\mathcal{I}(\Theta;\widetilde{\Theta})$. Because
	$\widetilde{\Theta}$ meets the distortion constraint, it is feasible for the infimum in
	\eqref{eq:rd_sensing}, whence
	$\mathcal{I}(\Theta;\widetilde{\Theta})\ge\mathcal{R}_{\Theta}(\epsilon_s)=b_s$.
\end{IEEEproof}

\begin{remark}
	{\em Proposition~\ref{prop1} is a converse statement: it does \emph{not} require the
		transmitter to quantize $\Theta$ into an explicit $b_s$-bit index, nor does it
		presuppose source-channel separation. It only says that \emph{any} architecture meeting
		the sensing target must deliver at least $b_s$ bits of information about $\Theta$
		through the channel.}
\end{remark}

As a running example, if $\Theta\sim\mathcal{N}(0,\sigma_{\theta}^2 I_d)$, the classical
quadratic-Gaussian rate-distortion function~\cite{berger1971,cover2006} gives the closed form
\begin{equation}
	b_s
	=
	\mathcal{R}_{\Theta}(\epsilon_s)
	=
	\frac{d}{2}\,
	\log_2^+\!
	\left(
	\frac{d\sigma_{\theta}^2}{\epsilon_s}
	\right),
	\label{eq:gaussian_rd}
\end{equation}
where $\log_2^+(x)=\max\{\log_2 x,0\}$. As expected, a more stringent sensing accuracy
(smaller $\epsilon_s$) translates into a larger sensing-equivalent payload $b_s$, i.e., a
heavier burden on the joint decoder.

\paragraph{Additivity of Communication and Sensing Information}
We now show that the two information requirements do not interfere with each other: under
the independence $w\perp\!\!\!\perp\Theta$, the receiver must extract $b$ bits for the
message \emph{and} $b_s$ bits for the sensing parameter, and these requirements are
\emph{additive}. Intuitively, knowing the message tells the receiver nothing about the
sensing parameter, so neither task can ``subsidize'' the other. Note that the lossless
message itself fits the same rate-distortion language: viewing $w$ as a uniform discrete
source with zero-distortion requirement, $\mathcal{R}_{w}(0)=\mathcal{H}(w)=b$.

\begin{proposition}\label{prop2}
	Suppose $w\perp\!\!\!\perp\Theta$. If a joint decoder recovers $w$ with error
	probability $p_e$ and estimates $\Theta$ with mean-squared error at most $\epsilon_s$,
	then
	\begin{equation}
		\mathcal{I}(w,\Theta;Y)
		\ge
		b+b_s-\mathcal{H}_2(p_e)-p_e b,
		\label{eq:joint_lower_bound}
	\end{equation}
	where $\mathcal{H}_2(\cdot)$ is the binary entropy function; in particular,
	$\mathcal{I}(w,\Theta;Y)\ge b+b_s-o(1)$ as $p_e\to0$.
\end{proposition}

\begin{IEEEproof}
	By the chain rule,
	$\mathcal{I}(w,\Theta;Y)=\mathcal{I}(w;Y)+\mathcal{I}(\Theta;Y|w)$, and we bound the
	two terms separately.
	\emph{Communication term.} Fano's inequality~\cite{fano1961,cover2006} yields
	$\mathcal{H}(w|Y)\le\mathcal{H}_2(p_e)+p_e b$, and since $\mathcal{H}(w)=b$,
	\begin{equation}
		\mathcal{I}(w;Y)\ge b-\mathcal{H}_2(p_e)-p_e b.
		\label{eq:comm_lower_bound}
	\end{equation}
	
	\emph{Sensing term.} Let $\widetilde{\Theta}$ be the decoder's sensing estimate. Since
	$\widetilde{\Theta}$ is a function of $Y$, the conditional data-processing inequality
	gives $\mathcal{I}(\Theta;Y|w)\ge\mathcal{I}(\Theta;\widetilde{\Theta}|w)$. The key
	observation is that, by $w\perp\!\!\!\perp\Theta$, conditioning on $w=a$ does not alter
	the distribution of $\Theta$, so the rate-distortion function
	$\mathcal{R}_{\Theta}(\cdot)$ applies unchanged under each conditioning event. Writing
	$d_a=\mathbb{E}[\|\Theta-\widetilde{\Theta}\|_2^2\,|\,w=a]$, the distortion constraint
	reads $\mathbb{E}_w[d_w]\le\epsilon_s$, and therefore
	\begin{equation}
		\mathcal{I}(\Theta;\widetilde{\Theta}|w)
		\ge
		\mathbb{E}_{w}\!\left[\mathcal{R}_{\Theta}(d_w)\right]
		\ge
		\mathcal{R}_{\Theta}\!\left(\mathbb{E}_{w}[d_w]\right)
		\ge
		\mathcal{R}_{\Theta}(\epsilon_s)=b_s,
		\label{eq:sensing_lower_bound}
	\end{equation}
	where the second inequality is Jensen's inequality applied to the convex function
	$\mathcal{R}_{\Theta}(\cdot)$, and the last holds because
	$\mathcal{R}_{\Theta}(\cdot)$ is non-increasing. Combining
	\eqref{eq:comm_lower_bound} and \eqref{eq:sensing_lower_bound} with the chain rule
	completes the proof.
\end{IEEEproof}

\begin{remark}
	{\em Proposition~\ref{prop2} is architecture-independent: it applies to arbitrary joint
		decoding rules and does not rely on source-channel separation. Its operational message
		is that, at the level of converse bounds, each active user effectively carries a
		payload of $b+b_s$ bits.}
\end{remark}

\paragraph{Equivalent Joint Rate-Distortion Interpretation}
The additivity in Proposition~\ref{prop2} admits a clean source-coding restatement.
Defining the joint source $S=(w,\Theta)$ with reconstruction
$\widetilde{S}=(\hat{w},\widetilde{\Theta})$, distortion measures
$\mathcal{D}_c(w,\hat{w})=\mathrm{1}\{w\neq\hat{w}\}$ and
$\mathcal{D}_s(\Theta,\widetilde{\Theta})=\|\Theta-\widetilde{\Theta}\|_2^2$, the joint
rate-distortion function factorizes:
\begin{equation}
	\mathcal{R}_{w,\Theta}(0,\epsilon_s)
	=
	\mathcal{R}_{w}(0)+\mathcal{R}_{\Theta}(\epsilon_s)
	=
	b+b_s.
	\label{eq:rd_additivity}
\end{equation}
\begin{IEEEproof}
The proof of \eqref{eq:rd_additivity} is standard and is only sketched here: the lower
bound follows from the chain rule and the data-processing inequality, mirroring the proof
of Proposition~\ref{prop2}. The upper bound follows by choosing the product reconstruction
law
$\mathbb{P}_{\hat{w},\widetilde{\Theta}|w,\Theta}
=\mathbb{P}_{\hat{w}|w}\,\mathbb{P}_{\widetilde{\Theta}|\Theta}$
with each factor achieving its individual rate-distortion function, which is feasible
precisely because $w\perp\!\!\!\perp\Theta$ and the two distortion constraints are
separable. We emphasize that \eqref{eq:rd_additivity} is conceptual: it characterizes the
minimum information rate compatible with both tasks, rather than mandating a separate
quantize-and-transmit implementation. With the effective payload $b+b_s$ in hand, the
next section quantifies the channel-side cost of delivering it over the unsourced MAC at
finite blocklength.
\end{IEEEproof}
\section{Achievability and Converse Bounds \\for the Standard $k_a$-User ISAC MAC}
\label{sec.standard_ka_jsac_mac}

This section quantifies the channel-side cost of delivering the effective payload
identified in the previous section over a standard $k_a$-user MAC~\cite{many_access},
in which each active user has its \emph{own} codebook. The development consists of
three parts: 
\begin{enumerate}
	\item A random-coding achievability bound
	(finite-blocklength form in \eqref{eq:finite_n_achievability_jsac}, asymptotic exponent
	in \eqref{eq:asymptotic_exponent_jsac}), in which the sensing requirement enters
	\emph{only} by enlarging the decoding search space from $2^b$ to $2^{b+b_s}$;
	\item Two converse bounds: One based on Fano's inequality and the MAC sum capacity, and the
	other on a genie-aided single-user reduction; and  \item Combination of derived bounds into a
	minimum-required energy-per-bit \eqref{eq:final_ebn0_lower_jsac} for a target pair
	$(\epsilon,\epsilon_s)$ of communication error and sensing distortion.
\end{enumerate}

Recall from \eqref{eq:bs_def} that $b_s=\mathcal{R}_{\Theta}(\epsilon_s)$, so the
effective payload per active user is $b_{\mathrm{e}}\triangleq b+b_s$, carried by the
effective index set $\mathcal{U}=\mathcal{W}\times\mathcal{V}_s$, where $\mathcal{V}_s$
is a sensing source-code index set with $|\mathcal{V}_s|=2^{b_s}$. Hence
\begin{equation}
	|\mathcal{U}|
	=
	2^{b+b_s}
	=
	2^{b_{\mathrm{e}}}.
	\label{eq:effective_search_space}
\end{equation}
(Integer-rounding effects in $2^{b_s}$ are ignored throughout; replacing $2^{b_s}$ by
$\lceil 2^{b_s}\rceil$ does not affect any asymptotic conclusion.) The key tension in the
ISAC setting is that the search space grows to $2^{b+b_s}$ while the physical
block-energy budget remains the communication budget $p$. To keep this distinction
explicit, let $p'\triangleq p/n$ denote the per-channel-use energy, and for the
random-coding ensemble choose
\begin{equation}
	0<\bar p<p,~
	\bar p'
	\triangleq
	\frac{\bar p}{n},
	\label{eq:pbar_prime_def}
\end{equation}
where $p,\bar p$ are block energies and $p',\bar p'$ their per-channel-use counterparts.
The effective system spectral efficiency is
\begin{equation}
	s_{\mathrm{e}}
	=
	\frac{k_a(b+b_s)}{n}.
	\label{eq:effective_spectral_efficiency}
\end{equation}
In the many-user asymptotic regime we set $k_a=\mu n$ with user density $\mu$, so that
$s_{\mathrm{e}}=\mu(b+b_s)$.

\paragraph{Random-Coding Achievability Bound}
For each user $k=1,\ldots,k_a$, generate an independent Gaussian codebook
$\mathcal{C}_k=\{C_k(u):u\in\mathcal{U}\}$ with
$C_k(u)\sim\mathcal{N}(0,\bar p'I_n)$. Since a Gaussian codeword does not automatically
satisfy the block-energy constraint, the transmitted vector is defined through
truncation:
\begin{equation}
	X_k(u)
	=
	\begin{cases}
		C_k(u),
		&
		\|C_k(u)\|_2^2\le p,\\[1mm]
		0,
		&
		\|C_k(u)\|_2^2>p,
	\end{cases}
	\label{eq:truncated_codeword_jsac}
\end{equation}
so that $\|X_k(u)\|_2^2\le p$ holds deterministically. The receiver applies the
nearest-neighbor decoder over the \emph{ordered} effective index tuple:
\begin{equation}
	(\hat u_1,\ldots,\hat u_{k_a})
	=
	\operatorname*{arg\,min}_{(\widetilde u_1,\ldots,\widetilde u_{k_a})\in\mathcal{U}^{k_a}}
	\left\|
	Y-\sum_{k=1}^{k_a}C_k(\widetilde u_k)
	\right\|_2^2.
	\label{eq:nn_decoder_standard_jsac}
\end{equation}
After $\hat u_k=(\hat w_k,\hat v_k)$ is decoded, $\hat w_k$ is the communication
estimate, and the sensing estimate is obtained by applying the sensing source decoder to
$\hat v_k$.

\paragraph{Energy-Truncation Error}
Because each user has its own codebook, there is no codeword-collision component; the
only measure-change penalty stems from truncation. Writing
$\|C_k(u_k)\|_2^2=\bar p'\sum_{i=1}^{n}g_i^2$ with $g_i\sim\mathcal{N}(0,1)$ i.i.d., the
event $\|C_k(u_k)\|_2^2>p$ is equivalent to $\frac{1}{n}\sum_{i=1}^{n}g_i^2>p/\bar p$,
and the union bound over the $k_a$ active users gives
\begin{equation}
	p_0^{(n)}
	\le
	k_a
	\mathbb{P}
	\left[
	\frac{1}{n}\sum_{i=1}^{n}g_i^2
	>
	\frac{p}{\bar p}
	\right]
	=
	k_a\,
	\frac{\Gamma\!\left(\frac{n}{2},\frac{np}{2\bar{p}}\right)}
	{\Gamma\!\left(\frac{n}{2}\right)},
	\label{eq:truncation_component_jsac}
\end{equation}
where the equality uses the fact that $\sum_{i=1}^{n}g_i^2$ follows a chi-square
distribution with $n$ degrees of freedom, and
$\Gamma(s,x)=\int_{x}^{\infty}t^{s-1}e^{-t}\mathrm{d}t$ is the upper incomplete Gamma
function.

\paragraph{Per-User Decoding Error}
Define the effective per-user decoding error probability \cite{MA1} as
\begin{equation}
	\epsilon_{\mathrm{e}}^{(n)}
	=
	\frac{1}{k_a}
	\sum_{k=1}^{k_a}
	\mathbb{P}[\hat u_k\neq u_k],
	\label{eq:effective_error_probability}
\end{equation}
indicating error when not all bits are restored, and, for $t=1,\ldots,k_a$, the event
\begin{equation}
	\mathcal{F}_t
	=
	\left\{
	\sum_{k=1}^{k_a}
	\mathrm{1}\{\hat u_k\neq u_k\}
	=
	t
	\right\},
	\label{eq:Ft_standard_jsac}
\end{equation}
so that, under the non-truncated Gaussian ensemble,
\begin{equation}
	\epsilon_{\mathrm{e}}^{(n)}
	=
	\sum_{t=1}^{k_a}
	\frac{t}{k_a}
	\mathbb{P}[\mathcal{F}_t];
	\label{eq:error_decomposition_standard_jsac}
\end{equation}
the truncation penalty $p_0^{(n)}$ will be added back at the end. To bound $\mathbb{P}[\mathcal{F}_t]$, fix an erroneous user set
$\mathcal{S}_0\subseteq\{1,\ldots,k_a\}$ with $|\mathcal{S}_0|=t$, and define the true
and false codeword sums
$A=\sum_{k\in\mathcal{S}_0}C_k(u_k)$ and
$B=\sum_{k\in\mathcal{S}_0}C_k(\widetilde u_k)$ with $\widetilde u_k\neq u_k$. By
independence of the codebooks, $A,B\sim\mathcal{N}(0,v_tI_n)$ with
$v_t=t\bar p'=t\bar p/n$. The pairwise nearest-neighbor error event is
\begin{equation}
	\|A-B+Z\|_2^2
	<
	\|Z\|_2^2.
	\label{eq:pairwise_error_jsac}
\end{equation}
Conditioning on $(A,Z)$ and applying a Chernoff bound with parameter $\lambda>0$,
followed by averaging over $B\sim\mathcal{N}(0,v_tI_n)$, yields
\begin{equation}
	\begin{aligned}
		&\mathbb{P}
		\left[
		\|A-B+Z\|_2^2<\|Z\|_2^2
		\mid A,Z
		\right] \\
		&\le
		(1+2v_t\lambda)^{-n/2}
		\exp
		\left\{
		\lambda\|Z\|_2^2
		-
		\frac{\lambda}{1+2v_t\lambda}
		\|A+Z\|_2^2
		\right\}.
	\end{aligned}
	\label{eq:pairwise_chernoff_jsac}
\end{equation}
For a fixed $\mathcal{S}_0$, each erroneous user independently selects one wrong
effective index from its own codebook, so the number of false ordered choices is
\begin{equation}
	(2^{b+b_s}-1)^t
	\le
	2^{t(b+b_s)}.
	\label{eq:false_choices_jsac}
\end{equation}
{\em A note on units (bits or nats):} $b$ and $b_s$ are measured in \textbf{bits}, whereas all unsubscripted
logarithms in the exponent analysis are natural, so exponents and rates are in
\textbf{nats per channel use}. A factor $\log 2$ therefore appears whenever $b$ or $b_s$
enters the random-coding exponent and the converse bounds are stated entirely in base-2 and require no conversion. Define
\begin{equation}\label{eq:r1_r2_jsac}
	r_1^{(n)} = \frac{(b+b_s)\log 2}{n}, \quad
	r_2^{(n)}(t) = \frac{1}{n}\log\binom{k_a}{t},
\end{equation}
so that $tr_1^{(n)}$ is the exponent associated with the false ordered choices and
$r_2^{(n)}(t)$ that of selecting the $t$ erroneous users. For $0\le\rho,\rho_1\le1$ and
$v>0$, define
\begin{subequations}
	\begin{align}
		\nu
		&=
		\frac{\rho\lambda}{1+2v\lambda}, \\
		a(\rho,\lambda;v)
		&=
		\frac{\rho}{2}\log(1+2v\lambda)
		+
		\frac{1}{2}\log(1+2v\nu),
		\label{eq:a_def_jsac}\\
		c(\rho,\lambda;v)
		&=
		\rho\lambda
		-
		\frac{\nu}{1+2v\nu},
		\label{eq:c_def_jsac}
	\end{align}
\end{subequations}
and the Gaussian distance exponent
\begin{equation}
	\begin{aligned}
		&\mathcal{E}_0(\rho,\rho_1;v)
		=\\
		&\max_{\lambda>0:\;1-2\rho_1c(\rho,\lambda;v)>0}
		\left\{
		\rho_1a(\rho,\lambda;v)
		+
		\frac{1}{2}
		\log
		\left(
		1-2\rho_1c(\rho,\lambda;v)
		\right)
		\right\},
	\end{aligned}
	\label{eq:E0_jsac}
\end{equation}
whose maximizing $\lambda$ admits the closed form
\begin{equation}
	\lambda^{\star}
	=
	\frac{
		v-1+
		\sqrt{
			(v-1)^2
			+
			4v
			\frac{1+\rho\rho_1}{1+\rho}
		}
	}
	{
		4v(1+\rho\rho_1)
	}.
	\label{eq:lambda_star_jsac}
\end{equation}
Applying Gallager's $\rho$-trick to the false ordered choices and the $\rho_1$-trick to
the erroneous user sets, exactly as in \cite{MA1}, gives
\begin{equation}
	\mathbb{P}[\mathcal{F}_t]
	\le
	\exp
	\left\{
	-n\mathcal{E}_n(t;\bar p)
	\right\},
	\label{eq:Ft_bound_jsac}
\end{equation}
where
\begin{equation}
	\mathcal{E}_n(t;\bar p)
	=
	\max_{0\le\rho,\rho_1\le1}
	\left[
	-\rho\rho_1 t r_1^{(n)}
	-\rho_1r_2^{(n)}(t)
	+
	\mathcal{E}_0
	\left(
	\rho,\rho_1;
	\frac{t\bar p}{n}
	\right)
	\right].
	\label{eq:finite_n_exponent_jsac}
\end{equation}
Combining \eqref{eq:error_decomposition_standard_jsac},
\eqref{eq:truncation_component_jsac}, and \eqref{eq:Ft_bound_jsac} yields the finite-$n$
random-coding achievability bound
\begin{equation}
	\epsilon_{\mathrm{e}}^{(n)}
	\le
	\sum_{t=1}^{k_a}
	\frac{t}{k_a}
	\exp
	\left\{
	-n\mathcal{E}_n(t;\bar p)
	\right\}
	+
	p_0^{(n)}.
	\label{eq:finite_n_achievability_jsac}
\end{equation}

Passing to the many-user regime $k_a=\mu n$ with $t=\theta k_a$, $0\le\theta\le1$, the
quantities in \eqref{eq:r1_r2_jsac} converge as
\begin{subequations}
	\begin{align}
		t r_1^{(n)}
		&\longrightarrow
		\theta\mu(b+b_s)\log 2,
		\label{eq:r1_limit_jsac}\\
		r_2^{(n)}(t)
		&\longrightarrow
		\mu h(\theta),
		\label{eq:r2_limit_jsac}\\
		\frac{t\bar p}{n}
		&\longrightarrow
		\theta\mu\bar p,
		\label{eq:power_limit_jsac}
	\end{align}
\end{subequations}
where $h(\theta)=-\theta\log\theta-(1-\theta)\log(1-\theta)$ is the binary entropy
function in nats. The asymptotic random-coding exponent is therefore
\begin{equation}
	\begin{aligned}
		&\mathcal{E}(\theta;\bar p)
		=\\
		&\max_{0\le\rho,\rho_1\le1}
		\left[
		-\rho\rho_1\theta\mu(b+b_s)\log 2
		-\rho_1\mu h(\theta)
		+
		\mathcal{E}_0
		\left(
		\rho,\rho_1;
		\theta\mu\bar p
		\right)
		\right].
	\end{aligned}
	\label{eq:asymptotic_exponent_jsac}
\end{equation}
For a target effective per-user decoding error probability $\epsilon$, define
\begin{equation}
	\mathcal{E}_{\min}(\bar p,\epsilon)
	=
	\min_{\epsilon\le\theta\le1}
	\mathcal{E}(\theta;\bar p).
	\label{eq:Emin_jsac}
\end{equation}
If there exists $0<\bar p<p$ such that $\mathcal{E}_{\min}(\bar p,\epsilon)>0$, then
{\em the standard $k_a$-user ISAC MAC is asymptotically achievable with effective
	per-user decoding error probability no larger than $\epsilon$}.

\begin{remark}
	{\em In \eqref{eq:asymptotic_exponent_jsac}, the sensing requirement affects the
		achievable exponent only through the false-search-space term
		$\theta\mu(b+b_s)\log 2$, whereas the energy term inside $\mathcal{E}_0(\cdot)$
		remains $\theta\mu\bar p$, determined by the physical block energy alone. Sensing thus
		increases the number of effective alternatives to be decoded without providing
		additional transmit energy.}
\end{remark}

\paragraph{Converse Bound Based on Fano's Inequality and Sum Capacity}
Let $p_{\mathrm{e},k}=\mathbb{P}[\hat w_k\neq w_k]$ and assume
$\frac{1}{k_a}\sum_{k=1}^{k_a}p_{\mathrm{e},k}\le\epsilon$, together with the per-user
sensing constraint
$\mathbb{E}[\|\Theta_k-\widetilde{\Theta}_k\|_2^2]\le\epsilon_s$ for all $k$. Write
$w^{k_a}=(w_1,\ldots,w_{k_a})$ and $\Theta^{k_a}=(\Theta_1,\ldots,\Theta_{k_a})$. By the
chain rule,
\begin{equation}
	\mathcal{I}(w^{k_a},\Theta^{k_a};Y)
	=
	\mathcal{I}(w^{k_a};Y)
	+
	\mathcal{I}(\Theta^{k_a};Y|w^{k_a}).
	\label{eq:chain_converse_jsac}
\end{equation}
For the communication term, applying Fano's inequality user-by-user,
$\mathcal{H}(w_k|Y)\le\mathcal{H}_2(p_{\mathrm{e},k})+p_{\mathrm{e},k}b$, and then the
concavity of $\mathcal{H}_2(\cdot)$ together with
$\frac{1}{k_a}\sum_k p_{\mathrm{e},k}\le\epsilon$ gives
\begin{equation}
	\begin{aligned}
		\mathcal{I}(w^{k_a};Y)
		&\ge
		k_ab
		-
		\sum_{k=1}^{k_a}
		\left[
		\mathcal{H}_2(p_{\mathrm{e},k})
		+
		p_{\mathrm{e},k}b
		\right] \\
		&\ge
		k_a
		\left[
		(1-\epsilon)b
		-
		\mathcal{H}_2(\epsilon)
		\right].
	\end{aligned}
	\label{eq:comm_fano_lower_jsac}
\end{equation}
For the sensing term, the parameters $\Theta_k$ are i.i.d.\ across users and independent
of all messages, so conditioning on $w^{k_a}$ leaves their distribution unchanged;
repeating the conditional rate-distortion argument in the proof of
Proposition~\ref{prop2} for each user (via the chain rule over $k$) yields
\begin{equation}
	\mathcal{I}(\Theta^{k_a};Y|w^{k_a})
	\ge
	k_ab_s.
	\label{eq:sensing_info_lower_jsac}
\end{equation}
Combining \eqref{eq:chain_converse_jsac}--\eqref{eq:sensing_info_lower_jsac},
\begin{equation}
	\mathcal{I}(w^{k_a},\Theta^{k_a};Y)
	\ge
	k_a
	\left[
	(1-\epsilon)b
	+
	b_s
	-
	\mathcal{H}_2(\epsilon)
	\right].
	\label{eq:joint_info_lower_jsac}
\end{equation}
On the other hand, the data-processing inequality gives
$\mathcal{I}(w^{k_a},\Theta^{k_a};Y)\le\mathcal{I}(X_1,\ldots,X_{k_a};Y)$, and under the
per-channel-use energy constraint $p'=p/n$ the AWGN MAC sum mutual information is upper
bounded by
\begin{equation}
	\mathcal{I}(X_1,\ldots,X_{k_a};Y)
	\le
	\frac{n}{2}
	\log_2
	\left(
	1+\frac{k_ap}{n}
	\right).
	\label{eq:sum_capacity_jsac}
\end{equation}
With $k_a=\mu n$, \eqref{eq:joint_info_lower_jsac} and \eqref{eq:sum_capacity_jsac} give
\begin{equation}
	\mu
	\left[
	(1-\epsilon)b
	+
	b_s
	-
	\mathcal{H}_2(\epsilon)
	\right]
	\le
	\frac{1}{2}
	\log_2(1+\mu p),
	\label{eq:fano_sum_converse_jsac}
\end{equation}
or equivalently
\begin{equation}
	p
	\ge
	p_{\mathrm{fano}}
	\triangleq
	\frac{1}{\mu}
	\left[
	2^{
		2\mu
		\left[
		(1-\epsilon)b
		+
		b_s
		-
		\mathcal{H}_2(\epsilon)
		\right]
	}
	-
	1
	\right]_+.
	\label{eq:p_fano_lower_jsac}
\end{equation}

\begin{remark}
	{\em The Fano penalty applies only to the lossless message $w_k$ and hence multiplies
		$b$; the sensing part contributes its full $b_s$ bits through the rate-distortion
		requirement. As $\epsilon\to0$, \eqref{eq:fano_sum_converse_jsac} reduces to
		$\mu(b+b_s)\le\frac{1}{2}\log_2(1+\mu p)$, matching the effective-payload
		interpretation of Section~\ref{sec.sensing_bits}.}
\end{remark}
\begin{figure}[t!]
	\centering
	\includegraphics[width=\columnwidth]{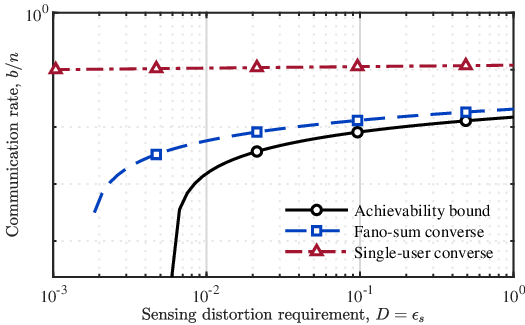}
	\caption{Communication--sensing tradeoff with $n=1000$, $k_a=100$,
		$\epsilon=0.1$, and $e_b/n_0=0$ dB. As the sensing distortion requirement
		$D=\epsilon_s$ tightens, the sensing-equivalent rate $b_s$ grows and the
		achievable communication rate $b/n$ shrinks accordingly, reflecting the
		effective payload relation $b_e=b+b_s$.}
	\label{fig:trade-offs}
		\vspace{-1mm}
\end{figure}
\paragraph{Converse Bound: Single-User Case}
A second converse follows from a genie-aided reduction. Since the average per-user error
probability is at most $\epsilon$, some user $j$ satisfies $p_{\mathrm{e},j}\le\epsilon$.
Revealing all other users' effective indices and transmitted vectors to the receiver,
their contributions can be subtracted from $Y$:
\begin{equation}
	\widetilde Y
	=
	Y-\sum_{k\neq j}X_k
	=
	X_j(w_j,\Theta_j)+Z,
	\label{eq:single_user_reduction_jsac}
\end{equation}
reducing user $j$ to a single-user real AWGN channel with block energy at most $p$,
over which $b+b_s$ effective bits must be conveyed. Since the genie can only help, any
feasible standard-MAC ISAC scheme must satisfy the single-user finite-energy converse:
with $\mathcal{Q}(x)=\mathbb{P}[g>x]$, $g\sim\mathcal{N}(0,1)$, and $\mathcal{Q}^{-1}$
its inverse,
\begin{equation}
	b+b_s
	\le
	-\log_2
	\mathcal{Q}
	\left(
	\sqrt{p}
	+
	\mathcal{Q}^{-1}(1-\epsilon)
	\right),
	\label{eq:single_user_energy_converse_jsac}
\end{equation}
or equivalently
\begin{equation}
	p
	\ge
	p_{\mathrm{energy}}
	\triangleq
	\left(
	\left[
	\mathcal{Q}^{-1}
	\left(
	2^{-(b+b_s)}
	\right)
	-
	\mathcal{Q}^{-1}
	\left(
	1-\epsilon
	\right)
	\right]_+
	\right)^2.
	\label{eq:p_energy_lower_jsac}
\end{equation}

\paragraph{Final Converse Bound}
Both converses must hold simultaneously, so any standard $k_a$-user ISAC MAC scheme
meeting the targets $(\epsilon,\epsilon_s)$ requires
$p\ge\max\{p_{\mathrm{fano}},p_{\mathrm{energy}}\}$ with $p_{\mathrm{fano}}$ and
$p_{\mathrm{energy}}$ given in \eqref{eq:p_fano_lower_jsac} and
\eqref{eq:p_energy_lower_jsac}. Since the physical budget is still the communication
block energy, the minimum-required energy-per-bit is
\begin{equation}
	\frac{e_b}{n_0}
	\ge
	\frac{1}{2b}
	\max
	\left\{
	p_{\mathrm{fano}},
	p_{\mathrm{energy}}
	\right\}.
	\label{eq:final_ebn0_lower_jsac}
\end{equation}

\begin{figure}[t!]
	\centering
	\includegraphics[width=\columnwidth]{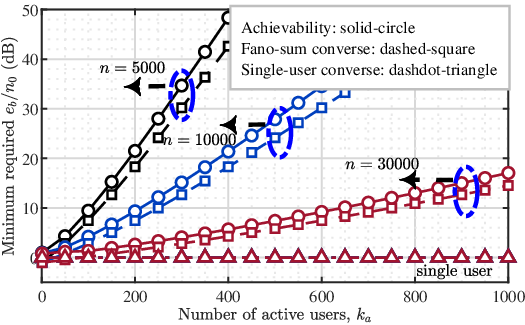}
	\caption{Minimum required $e_b/n_0$ (dB) versus the number of active users
		$k_a$ to meet the targets $(\epsilon,\epsilon_s)$ with $D=10^{-3}$,
		$b=100$, $\epsilon=0.1$, and $n\in\{5000,10^4,3\times10^4\}$. Larger
		blocklengths substantially reduce the energy requirement, and the Fano-sum
		converse closely tracks the multi-user loading trend while the
		single-user converse stays nearly flat.}
	\label{fig:minimum_required}
		\vspace{-1mm}
\end{figure}

\section{Numerical Results}\label{sec.4}
We evaluate the proposed ISAC bounds for a Gaussian remote source with $d=4$
and $\sigma_\theta^2=1$, whose sensing-equivalent rate $b_s$ is given by
\eqref{eq:gaussian_rd}. In all figures the energy efficiency is normalized to
the communication payload as $e_b/n_0 = p/(2b)$, and the sensing fidelity is
treated as a constraint folded into the effective payload $b_e=b+b_s$.

Fig.~\ref{fig:trade-offs} illustrates the communication--sensing tradeoff for
$n=1000$, $k_a=100$, $\epsilon=0.1$, and $e_b/n_0=0$ dB. As the prescribed
distortion $D$ decreases, $b_s$ grows according to the Gaussian rate-distortion
function, and the achievable communication rate $b/n$ decreases monotonically.
This directly reflects the relation $b_e=b+b_s$, where stricter sensing
fidelity consumes a larger share of the fixed effective information load.
\begin{figure}[t!]
	\centering
	\includegraphics[width=\columnwidth]{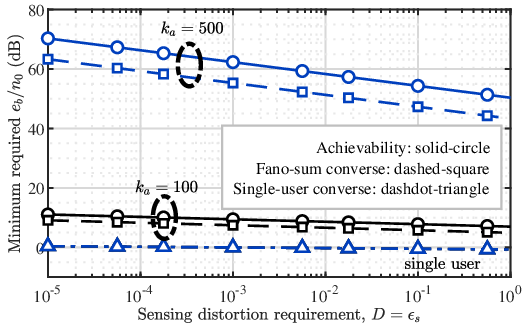}
	\caption{Minimum required $e_b/n_0$ (dB) versus the sensing distortion
		requirement $D$ with $n=5000$, $b=100$, $\epsilon=0.1$, and
		$k_a\in\{100,500\}$. The energy cost grows almost linearly per decade of
		distortion tightening, and the slope steepens with the user load,
		showing that multiple-access interference amplifies the energy price of
		sensing fidelity/distortion.}
	\label{fig:energy_distortion}
\end{figure}

Fig.~\ref{fig:minimum_required} plots the minimum required $e_b/n_0$ versus
the number of active users $k_a$ for $D=10^{-3}$, $b=100$, $\epsilon=0.1$, and
$n\in\{5000,10^4,3\times10^4\}$. The required energy increases monotonically
with $k_a$, while larger blocklengths substantially reduce it, and the
achievability curves lie strictly above the converses. The Fano-sum converse
tightly tracks the multi-user loading trend, whereas the genie-aided
single-user converse remains nearly independent of $k_a$, consistent with its
single-user origin.

Fig.~\ref{fig:energy_distortion} quantifies the energy price of sensing
fidelity by plotting the minimum required $e_b/n_0$ versus $D$ for $n=5000$,
$b=100$, $\epsilon=0.1$, and $k_a\in\{100,500\}$. Since $b_s$ grows
logarithmically in $1/D$, the required energy increases almost linearly per
decade of distortion tightening. The slope is strongly load-dependent.
Tightening $D$ from $1$ to $10^{-5}$ costs about $20$ dB for $k_a=500$ but
only about $4$ dB for $k_a=100$, showing that the energy cost of sensing
accuracy is amplified by multiple-access interference. The single-user
converse is nearly flat, which confirms that this amplification is a genuinely
multi-user effect.

Fig.~\ref{fig:joint_orthogonal} compares the proposed joint encoding, in which
a single codeword carries $b_e=b+b_s$ over the entire blocklength, against an
orthogonal two-phase baseline. In the baseline, the frame is split into a
communication phase of $\alpha n$ channel uses carrying $b$ bits and a sensing
report phase of $(1-\alpha)n$ channel uses carrying $b_s$ bits, each designed
for a target error $\epsilon/2$ so that the union bound meets the overall target $\epsilon$. The minimum energy of each phase is computed from the same
achievability bound, and the total energy is then minimized over the time split $\alpha$. Even against this optimized baseline, joint encoding saves
about $0.7$ dB at small $k_a$ and about $0.3$ dB at $k_a=300$. The gain stems
from coding over the full blocklength, which avoids the dispersion penalty of
two shortened sub-blocks, and from a single shared error budget instead of the
split $\epsilon/2$ targets. Both achievability curves remain above the
Fano-sum converse, and the joint scheme stays uniformly closer to it,
indicating that integrating sensing and communication into one codeword is
strictly more energy-efficient than time-sharing.
\begin{figure}[t!]
	\centering
	\includegraphics[width=\columnwidth]{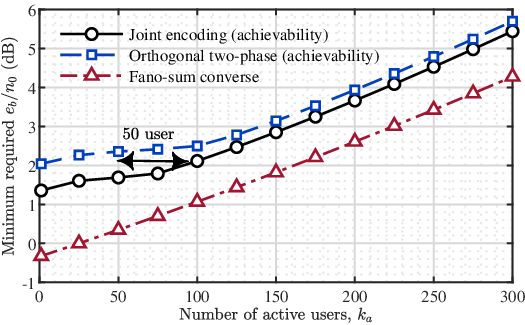}
	\caption{Joint encoding versus the optimized orthogonal two-phase
		baseline with $n=3\times10^4$, $b=100$, $D=10^{-6}$, and $\epsilon=0.1$.
		Carrying $b_e=b+b_s$ in a single codeword over the full blocklength
		consistently outperforms time-sharing even with an optimized split
		$\alpha$, and the joint scheme stays uniformly closer to the Fano-sum
		converse.}
	\label{fig:joint_orthogonal}
	\vspace{-1mm}
\end{figure}

\section{Conclusion}\label{sec.5}
We developed a finite-blocklength ISAC framework for Gaussian
multiple-access channels by proving that sensing-distortion constraints
are equivalent to source-coding requirements, which folds the fidelity/distortion
target into a sensing-equivalent payload and reduces ISAC to
recovering the effective load. We derived an achievability
bound, a Fano-sum many-user converse, and a genie-aided single-user
converse. Numerical results show that stricter sensing accuracy reduces
the achievable rate, the energy price of fidelity is amplified by
multiple-access interference, and joint encoding strictly outperforms
optimized orthogonal time-sharing, quantifying the integration gain of
ISAC in the short-packet regime.

\end{document}